\documentclass[letterpaper, 10 pt, conference]{ieeeconf} 

\IEEEoverridecommandlockouts                             
\usepackage{graphics} 
\usepackage{graphicx}
\usepackage{float}
\usepackage{epsfig} 
\usepackage{newtxmath} 
\usepackage{bbm} 
\usepackage{times} 
\usepackage{amsmath} 
\usepackage{amssymb}  
\usepackage{color}
\usepackage{mathrsfs}
\usepackage[margin=1in]{geometry}

\def\bs{\boldsymbol}
\def\inr{\in\mathbb{R}}
\newcommand{\norm}[1]{\left\lVert#1\right\rVert}

\makeatletter
\let\NAT@parse\undefined
\makeatother

\usepackage[square, sort&compress, numbers]{natbib}

\def\real{\mathbb{R}}

\newcommand\oprocendsymbol{\hbox{$\square$}}
\newcommand\oprocend{\relax\ifmmode\else\unskip\hfill\fi\oprocendsymbol}

\newtheorem{theorem}{Theorem}

\newtheorem{lemma}{Lemma}
\newtheorem{corollary}{Corollary}

\newtheorem{remark}{Remark}

\newtheorem{assumption}{Assumption}

\renewcommand{\baselinestretch}{0.999}

\renewcommand{\theenumi}{\roman{enumi}}

\newcommand\bit[1]{\textit{\textbf{#1}}}

\begin{document}
%
\title{Robust tracking of an unknown trajectory with a multi-rotor UAV: A high-gain observer approach}
%
%
%

\author{Connor J. Boss,
        Vaibhav Srivastava, and
        Hassan K. Khalil
\thanks{C. J. Boss, V. Srivastava, and H. K. Khalil are with the Department of Electrical and Computer Engineering, Michigan State University, East Lansing, MI, 48823 USA e-mail: \texttt{\{bossconn, vaibhav, khalil\}@egr.msu.edu}}}
\maketitle

\begin{abstract}
We study a trajectory tracking problem for a multi-rotor in the presence of modeling error and external disturbances. The desired trajectory is unknown and generated from a reference system with unknown or partially known dynamics. We assume that only position and orientation measurements for the multi-rotor and position measurements for the reference system can be accessed. We adopt an extended high-gain observer (EHGO) estimation framework to estimate the feed-forward term required for trajectory tracking, the multi-rotor states, as well as modeling error and external disturbances. We design an output feedback controller for trajectory tracking that comprises a feedback linearizing controller and the EHGO. 
We rigorously analyze the proposed controller and establish its stability properties. Finally, we numerically illustrate our theoretical results using the example of a multi-rotor landing on a ground vehicle.  
\end{abstract}


%
\IEEEpeerreviewmaketitle

\section{Introduction}
We envision a future in which aerial vehicles provide services such as parcel delivery, remote monitoring, and maintenance. In such scenarios, a ground-based vehicle may provide charging services for multi-rotors or act as a staging area for storing parcels. 
In these cases, the multi-rotors will be required to make multiple landings on the ground vehicle, which need not be stationary.

One of the principle challenges in achieving a landing on a moving vehicle is the generation of a trajectory. The multi-rotor and ground vehicle 
may be operated by different service-providers, which may prohibit communication between the vehicles due to security reasons, or if the two have incompatible communication equipment. In this case, the trajectory of the ground vehicle will need to be inferred by the multi-rotor. 

The multi-rotor may only be able to measure the position of the ground vehicle in real-time. However, to achieve efficient control performance, higher-order derivatives of the trajectory are required. Furthermore, the modeling error and external disturbances---such as wind gusts, the ground effect, aerodynamic drag, parcel size and weight, or items shifting inside a parcel---may also deteriorate the multi-rotor's tracking performance resulting in a poor landing. In this work, we study this problem and design an output feedback controller that addresses these challenges.

From control design to path planning and disturbance rejection, much work has been devoted to studying multi-rotor UAV control design; see~\cite{papachristos2018modeling,kumar2012opportunities} for a survey. 
The problem of autonomous landing of a multi-rotor on a mobile platform has also received some attention~\cite{kong2014vision,gautam2014survey}. 
Many control methodologies have been applied to landing on a mobile platform, including model predictive control \cite{feng2018autonomous,maces2017autonomous}, PI control \cite{herisse2011landing,serra2016landing,sanchez2014approach}, and feedback linearizing control \cite{hoang2017vision}. 

State estimators such as a Kalman filter have been used to estimate the dynamic state of the mobile platform \cite{sanchez2014approach,kim2014outdoor} under the assumptions that the dynamic model of the mobile platform is known and it travels with unknown constant velocity. Through our EHGO design, these assumptions are relaxed, requiring no information about the mobile platform's dynamics or input. An alternative approach to estimate the state of the mobile platform uses optical flow data \cite{herisse2011landing,serra2016landing}, or visual cue data \cite{kendoul2014four} in which a dynamic model of the mobile platform is not required. In these cases the relative velocity is estimated through the optical flow algorithms and is minimized in the control to ensure tracking.


Many of the approaches in the literature either do not consider modeling error and external disturbances, or consider them to be constant or slowly time-varying \cite{herisse2011landing,serra2016landing}. 
In contrast, our approach only requires that the disturbance be bounded and continuously differentiable.

In this paper, we study a trajectory tracking problem for a multi-rotor in the presence of modeling error and external disturbances. The desired trajectory is unknown and generated from a reference system with unknown or partially known dynamics. We assume that only position and orientation measurements for the multi-rotor and position measurements for the reference system can be accessed.  We design an output feedback controller that robustly tracks such unknown desired trajectories. 
The contributions of this work are as follows:
\begin{itemize}
    \item We design and rigorously analyze an EHGO to estimate modeling error and external disturbances, feed-forward control for trajectory tracking, and multi-rotor states for output feedback control.
    \item We design and analyze a robust feedback linearizing controller that mitigates modeling errors and external disturbances using their estimates. 
    \item We rigorously characterize the stability of the overall output feedback system. 
    \item We illustrate the effectiveness of our output feedback controller through simulation using the example of a multi-rotor landing on a mobile platform.
\end{itemize}

The remainder of the paper is organized as follows. The system dynamics and control are introduced in Section \ref{sec:systemDynamics} and \ref{sec:ControlDesign}, respectively. The controller is analyzed in Section \ref{sec:stabilityAnalysis}, and is validated through simulation in 
Section \ref{sec:simulation}. Conclusions are presented in Section \ref{sec:conclusion}.

\section{System Dynamics}\label{sec:systemDynamics}
A multi-rotor UAV is an underactuated mechanical system. While there can be $n\in\{4,6,8,...\}$ rotors, only four degrees of freedom can be controlled in the classic configuration where all rotors are co-planar. To handle the underactuation, as discussed below,  the rotational dynamics are controlled to create a virtual control input for the translational dynamics.


\subsection{Rotational Dynamics}
The rotational dynamics of the multi-rotor are \cite{lee2010geometric} 
\begin{equation}\label{eq:rotational-dynamics}
    \bs{\tau} = J\dot{\Omega} + \Omega \times J\Omega,
\end{equation}
where $J\inr^{3\times3}$ is the inertia matrix, $\bs{\tau}\inr^3$ is the torque applied to the multi-rotor and $\Omega\inr^3$ is the angular velocity, each expressed in the body-fixed frame. 

Consider the orientation of the multi-rotor expressed in terms on Euler angles $\bs{\theta}_1 = [\phi \ \theta \ \psi]^\top\in (-\frac{\pi}{2},\frac{\pi}{2})^2 \times (-\pi,\pi]$. The angular velocity $\Omega$ is related to the Euler angle rates $\bs{\theta}_2 = [\dot{\phi} \ \dot{\theta} \ \dot{\psi}]^\top \inr^3$ in the inertial frame as:
\begin{equation*}
    \bs{\theta}_2 = \Psi\Omega, \quad \Psi = \left[\begin{matrix}
    1&s_\phi t_\theta&c_\phi t_\theta\\
    0&c_\phi&-s_\phi\\
    0&s_\phi/c_\theta&c_\phi/c_\theta \end{matrix}\right], \quad \Omega = \Psi^{-1}\bs{\theta}_2,
\end{equation*}
where $c_{(\cdot)},s_{(\cdot)},t_{(\cdot)}$ denote $\cos(\cdot)$, $\sin(\cdot)$, $\tan(\cdot)$, respectively. The rotational dynamics can be equivalently written in terms of Euler angles:
\begin{equation}\label{eq:rotationalDynamics}
    \begin{split}
        \dot{\bs{\theta}}_1 &= \bs{\theta}_2,\\
        \dot{\bs{\theta}}_2 &= \dot{\Psi} \Psi^{-1} \bs{\theta}_2 - \Psi J^{-1}(\Psi^{-1} \bs{\theta}_2 \times J\Psi^{-1} \bs{\theta}_2)\\
        & \quad + \Psi J^{-1}\bs{\tau} + \bs{\sigma}_\xi,\\
    \end{split}
\end{equation}
where $\bs{\sigma}_\xi \inr^3$ is an added term to represent the lumped rotational disturbance and satisfies \textit{Assumption \ref{as:disturbance}}.

\begin{assumption}[Properties of Disturbances]\label{as:disturbance}
For a control system with state $\bs x \in \real^n$, expressed in lower triangular form, such as~\eqref{eq:rotationalDynamics}, any disturbance term is assumed to enter only $x_n$ dynamics. The disturbance term  is also assumed to be continuously differentiable and its partial derivatives with respect to states are bounded on compact sets of those states for all $t\geq0$.
\end{assumption}

Let $\bs{\theta}_r = [\phi_r \ \theta_r \ \psi_r]^\top \in (-\frac{\pi}{2},\frac{\pi}{2})^2 \times (-\pi,\pi]$ be the rotational reference signal. Define rotational tracking error variables
\begin{equation*}
    \begin{gathered}
        \bs{\xi}_1 = \bs{\theta}_1 - \bs{\theta}_r, \quad \bs \xi_2 = \dot{\bs{\xi}_1} = \bs{\theta}_2 - \dot{\bs{\theta}}_r, \quad \bs \xi = [\bs \xi_1^\top \; \bs \xi_2^\top]^\top.
    \end{gathered}
\end{equation*}
Then, the rotational dynamics can be written in terms of tracking error
\begin{equation}\label{eq:rotationalErrorDynamics}
    \begin{split}
        \dot{\bs{\xi}}_1 &= \bs{\xi}_2, \\
        \dot{\bs{\xi}}_2 &= f(\bs{\xi},\bs{\theta}_1,\dot{\bs{\theta}}_r) + G(\bs{\theta}_1)\bs{\tau} + \bs{\sigma}_\xi - \ddot{\bs{\theta}}_r,
    \end{split}
\end{equation}
where
\begin{equation*}
    \begin{split}
        f(\bs{\xi},\bs{\theta}_1,\dot{\bs{\theta}}_r) &= \dot{\Psi}\Psi^{-1}(\bs{\xi}_2 + \dot{\bs{\theta}}_r) \\
        &\quad -\Psi J^{-1}(\Psi^{-1}(\bs{\xi}_2 + \dot{\bs{\theta}}_r) \times J\Psi^{-1}(\bs{\xi}_2 + \dot{\bs{\theta}}_r)), \\
        G(\bs{\theta}_1) &= \Psi J^{-1}.
    \end{split}
\end{equation*}
Suppose that only $\dot{\bar{\bs{\theta}}}_r$, an estimate of $\dot{\bs{\theta}}_r$, is known. Then \eqref{eq:rotationalErrorDynamics} can be rewritten as
\begin{equation}\label{eq:rotationalErrorDynamicsFinal}
    \begin{split}
        \dot{\bs{\xi}}_1 &= \bs{\xi}_2, \\
        \dot{\bs{\xi}}_2 &= f(\bs{\xi},\bs{\theta}_1,\dot{\bar{\bs{\theta}}}_r) + G(\bs{\theta}_1)\bs{\tau} + \bs{\varsigma}_\xi,
    \end{split}
\end{equation}
where $\bs{\varsigma}_\xi = \bs{\sigma}_\xi - \ddot{\bs{\theta}}_r + [f(\bs{\xi},\bs{\theta}_1,\dot{\bs{\theta}}_r) - f(\bs{\xi},\bs{\theta}_1,\dot{\bar{\bs{\theta}}}_r)]$, which also satisfies \textit{Assumption \ref{as:disturbance}} based on the properties of $f$ and by assuming the reference trajectory is continuously differentiable.

\subsection{Translational Dynamics}

The translational dynamics of the multi-rotor are \cite{lee2010geometric}
\begin{equation}\label{eq:translationalDynamics}
    \begin{split}
        \dot{\bs{p}}_1 &= \bs{p}_2, \\
        \dot{\bs{p}}_2 &= -\frac{u_f}{m}R_3(\bs{\theta}_1) + g\bs{e}_z + \bs{\sigma}_\rho.
    \end{split}
\end{equation}
Here, $\bs{p}_1 = [x \ y \ z]^\top\inr^3$ is the position of the center of mass of the aerial platform in the inertial frame, 
$u_f = \sum_{i=1}^n{\bar{f}_i} \inr$ is the cumulative thrust force, $\bar{f}_i \inr$ is the force generated by the $i$-th rotor, $m\inr$ is the mass of the aerial platform, $g$ is the gravitational constant, $\bs{e}_z = [0 \ 0 \ 1]^\top$, and $\bs{\sigma}_\rho \inr^3$ is the lumped translational disturbance term which satisfies \textit{Assumption \ref{as:disturbance}}, and $R_3(\bs{\theta}_1)\inr^3$ is 
\begin{equation*}
    R_3(\bs{\theta}_1) = \left[\begin{matrix}
    c_\phi s_\theta c_\psi + s_\phi s_\psi\\
    c_\phi s_\theta s_\psi - s_\phi c_\psi\\
    c_\phi c_\theta \end{matrix}\right].
\end{equation*}

Let $\bs{p}_r = [x_r \ y_r \ z_r]^\top \inr^3$ be the translational reference signal. Define translational error variables
\begin{equation*}
    \begin{gathered}
        \bs{\rho}_1 = \bs{p}_1 - \bs{p}_r, \quad \bs{\rho}_2 = \dot{\bs{\rho}_1} = \bs{p}_2 - \dot{\bs{p}}_r, \quad \bs \rho=[\bs \rho_1^\top \; \bs \rho_2^\top]^\top.
    \end{gathered}
\end{equation*}
Then, the translational dynamics can be written in terms of tracking error as
\begin{equation}\label{eq:translationalErrorDynamics}
    \begin{split}
        \dot{\bs{\rho}}_1 &= \bs{\rho}_2, \\
        \dot{\bs{\rho}}_2 &= -\frac{u_f}{m}R_3(\bs{\theta}_1) + g\bs{e}_z + \bs{\sigma}_\rho - \ddot{\bs{p}}_r.
    \end{split}
\end{equation}

\subsection{Reference System Dynamics}
We assume that the reference trajectory that the multi-rotor UAV will track is generated by the system
\begin{equation}\label{eq:groundVehicleDynamics}
    \begin{split}
        \dot{\bs{x}}_{c_1} &= \bs{x}_{c_2}, \\
        \dot{\bs{x}}_{c_2} &= f_c(\bs{x}_c,\bs{u}_c), \\
    \end{split}
\end{equation}
where $\bs{x}_{c_1} = [x_c \ y_c \ z_c]^\top \inr^3$ is the position of the reference system, $\bs x_c =[\bs x_{c_1}^\top, \bs x_{c_2}^\top]^\top$ is the system state, $\bs u_c$ is the unknown system input, and $f_c(\bs{x}_c,\bs{u}_c)$ is some unknown function. We take system input $\bs{u}_c = g_c(t,\bs{x}_c)$ and let $\bar{f}_c(t,\bs{x}_c) = f_c(\bs{x}_c,\bs{u}_c)$.  We assume that $\frac{\partial \bar{f}_c(t,\bs{x}_c)}{\partial \bs{x}_c}\dot{\bs{x}}_c$ satisfies \textit{Assumption \ref{as:disturbance}}.

\section{Control Design} \label{sec:ControlDesign}
In this section, we first design a trajectory tracking feedback linearizing controller for the rotational system and subsequently use the rotational trajectory to design a trajectory tracking controller for the translational system. 


\subsection{Rotational Control}
The rotational control feedback linearizes the rotational tracking error dynamics \eqref{eq:rotationalErrorDynamics} as follows
\begin{equation}\label{eq:rotationalControl}
        \bs{\tau}_d = G^{-1}(\bs{\theta}_1)[\bs{f}_r - f(\bs{\xi},\bs{\theta}_1,\dot{\bar{\bs{\theta}}}_r)],
\end{equation}
where $\bs{f}_r = -\beta_1\bs{\xi}_1 - \beta_2\bs{\xi}_2 - \bs{\varsigma}_\xi$, and $\beta_1,\beta_2\inr_{>0}$ are constant gains. This results in the following closed-loop rotational tracking error system
\begin{equation}\label{eq:rotationalClosedLoopStateFeedback}
    \begin{split}
        \dot{\bs{\xi}}_1 &= \bs{\xi}_2,\\
        \dot{\bs{\xi}}_2 &= -\beta_1\bs{\xi}_1 - \beta_2\bs{\xi}_2.
    \end{split}
\end{equation}

\subsection{Translational Control}
The translational control uses the total thrust, $u_f$, as the direct control input and the desired roll and pitch trajectories, $\phi_r$ and $\theta_r$, as virtual control inputs. The translational control will be designed in view of the potential tracking errors in roll and pitch trajectories, leading to the following modification of the translational error dynamics \eqref{eq:translationalErrorDynamics}
\begin{equation}
    \begin{split}
        \dot{\bs{\rho}}_1 &= \bs{\rho}_2, \\
        \dot{\bs{\rho}}_2 &= -\frac{u_f}{m}R_3(\bs{\theta}_r + \bs{\xi}_1) + g\bs{e}_z + \bs{\sigma}_\rho - \ddot{\bs{p}}_r.
    \end{split}
\end{equation}
The translational subsystem dynamics can be redefined in terms of the nominal translational model with the addition of an error term as follows
\begin{equation}\label{eq:translationalDynamicsWithErrorTerm}
    \begin{split}
        \dot{\bs{\rho}}_1 &= \bs{\rho}_2, \\
        \dot{\bs{\rho}}_2 &= -\frac{u_f}{m}R_3(\bs{\theta}_r) + g\bs{e}_z + \bs{\sigma}_\rho - \ddot{\bs{p}}_r + \bs{e}_\theta(t,\bs{\xi}_1),
    \end{split}
\end{equation}
where
\begin{equation*}
    \bs{e}_\theta(t,\bs{\xi}_1) = -\frac{u_f}{m}(R_3(\bs{\theta}_r + \bs{\xi}_1) - R_3(\bs{\theta}_r)).
\end{equation*}
In the case of perfect rotational tracking, i.e., $e_\theta(t,\bs{\xi}_1) = 0$, the system \eqref{eq:translationalDynamicsWithErrorTerm} can be feedback linearized in the following manner
\begin{equation}\label{eq:translationalControl}
    \begin{split}
        \phi_r &= \tan^{-1}\left( \dfrac{-f_y}{\sqrt{f_x^2 + (f_z-g)^2}} \right), \quad \psi_r = 0,\\ 
        \theta_r &= \tan^{-1}\left(\dfrac{f_x}{f_z-g}\right), \quad
        u_{f d} = -\dfrac{m(f_z-g)}{c_{\phi_r} c_{\theta_r}},
    \end{split}
\end{equation}
where $\bs{f}_t = [f_x \ f_y \ f_z]^\top \inr^3$ is the forcing function, defined as follows
\begin{equation}\label{eq:translationalForcing}
    \bs{f}_t = -\gamma_1\bs{\rho}_1 - \gamma_2\bs{\rho}_2 - \bs{\sigma}_\rho + \ddot{\bs{p}}_r,
\end{equation}
for constants $\gamma_1,\gamma_2 \inr_{>0}$. This leads to the following closed-loop translational subsystem with the inclusion of tracking error from the rotational subsystem
\begin{equation}\label{eq:translationalClosedLoopStateFeedback}
    \begin{split}
        \dot{\bs{\rho}}_1 &= \bs{\rho}_2, \\
        \dot{\bs{\rho}}_2 &= -\gamma_1\bs{\rho}_1 - \gamma_2\bs{\rho}_2 + \bs{e}_\theta(t,\bs{\xi}_1).
    \end{split}
\end{equation}
Note that the controller \eqref{eq:rotationalControl} requires the estimate $\dot{\bar{\bs{\theta}}}_r$, however, only $\bs{\theta}_r$ is given by the translational controller. The derivative of the reference trajectory $\dot{\bs{\theta}}_r$ can be computed analytically from the translational controller as
\begin{equation*}
    \begin{split}
        \dot{\phi}_r &= \frac{ f_y\left(\dot{f}_x f_x + \dot{f}_z(f_z - g)\right) - \dot{f}_y\left(f_x^2 + (f_z - g)^2\right)}{\left(f_x^2+(f_z - g)^2\right)^{1/2}\left(f_x^2 + f_y^2 + (f_z - g)^2\right)},\\
        \dot{\theta}_r &= \frac{\dot{f}_x(f_z - g) - f_x \dot{f}_z}{f_x^2 + (f_z - g)^2}, \\
        \dot{\psi}_r &= 0,
    \end{split}
\end{equation*}
where $\dot{\bs{f}}_t = \left[\dot{f}_x \ \dot{f}_y \ \dot{f}_z\right]^\top$ and
\begin{equation}\label{eq:f_tDot}
    \begin{split}
    \dot{\bs{f}}_t &= -\gamma_1\bs{\rho}_2 - \gamma_2\left[-\frac{u_f}{m}R_3(\bs{\theta}_r) + g\bs{e}_z + \bs{\sigma}_\rho - \ddot{\bs{p}}_r \right] \\
    &\quad - \dot{\bs{\sigma}}_\rho + \bs{p}_r^{(3)}.
    \end{split}
\end{equation}
Our estimate $\dot{\bar{\bs{\theta}}}_r$ is obtained by setting $\dot{\bs{\sigma}}_\rho = 0$ in the expression for $\dot{\bs{\theta}}_r$.
While the substitution \eqref{eq:f_tDot} requires the third order derivative of the translational reference, it is shown in the EHGO design that the translational reference must be fifth order differentiable. 

\subsection{Extended High-Gain Observer Design}\label{sec:EHGODesign}
A multi-input multi-output EHGO is designed similar to \cite{lee2012control,lee2016output} to estimate higher-order states of the error dynamic systems \eqref{eq:rotationalErrorDynamics} and \eqref{eq:translationalErrorDynamics}, uncertainties arising from modeling error and external disturbances, as well as the reference trajectory based on the reference system dynamics \eqref{eq:groundVehicleDynamics}. It is shown in \cite{boss2017uncertainty} that it is necessary to include actuator dynamics in the multi-rotor model for EHGO design. The actuator dynamics reside in the same time-scale as the EHGO, and therefore can not be ignored in the EHGO dynamics.

The actuators used on multi-rotor UAVs are Brushless DC (BLDC) motors, which require electronic speed controllers. 
These controllers introduce dynamic delays \cite{franchi2017adaptive} of the following form
\begin{equation}
    \tau_m \dot{\bs{\omega}} = (\bs{\omega}_\text{des} - \bs{\omega}),
\end{equation}
where $\tau_m \inr$ is the time constant of the actuator system, $\bs{\omega} \inr^n$ is a vector of angular rates of the rotors and $\bs{\omega}_\text{des} \inr^n$ is a vector of rotor angular rate control inputs. Since feedback of the rotor angular rate is not available, it can be simulated by the following system
\begin{equation}
    \tau_m\dot{\hat{\bs{\omega}}} = (\bs{\omega}_\text{des} - \hat{\bs{\omega}}),
\end{equation}
where $\hat{\bs{\omega}}\inr^n$ is a vector of simulated rotor angular rates.
The dynamics take total thrust force, $u_f$, and body-fixed torques, $\bs{\tau}$, as inputs.
The thrust and torques are generated by applying different forces with each actuator, which is a function of the rotor angular rate
\begin{equation}
    \bar{f}_i = b\omega_i^2, \quad \text{for} \ i = \{1,\dots,n\},
\end{equation}
where $b \inr$ is a constant relating angular rate to force and $\omega_i \inr$ is the $i$-th rotor angular rate. These individual actuator forces are then mapped through a matrix, $M \inr^{4\times n}$, based on the geometry of the multi-rotor aerial platform, allowing the squared rotor angular rates to be taken as the control input to the model
\begin{equation}\label{eq:inputTransformation}
    \left[\begin{matrix} u_f \\ \bs{\tau} \end{matrix}\right]= b M\bs{\omega}_s, \quad \bs{\omega}_s = \left[\omega_1^2 \dots \omega_n^2 \right]^\top.
\end{equation}
The inverse operation is used to generate the desired rotor speeds $\bs{\omega}_\text{des}$ from the feedback linearizing control signals $u_{f d}$ and $\bs{\tau}_d$. For $n>4$, the inverse of \eqref{eq:inputTransformation} is an over-determined system which admits infinitely many solutions. In this case, we focus on the minimum energy solution, in which the pseudo-inverse of $M$ is used.

The dynamics \eqref{eq:rotationalDynamics}, \eqref{eq:translationalDynamics}, and \eqref{eq:groundVehicleDynamics} can be combined into one set of equations for the observer where the state space will be extended to estimate disturbance terms. Since the third derivative of the reference trajectory is required, the reference system's dynamics will be extended to include the third derivative of its position.
\begin{equation}
    \begin{split}
    \dot{\bs{\rho}}_1 &= \bs{\rho}_2, \\
        \dot{\bs{\rho}}_2 &= -\frac{u_f}{m}R_3(\bs{\theta}_1) + g\bs{e}_z + \bs{\sigma}_\rho - \ddot{\bs{p}}_r, \\
        \dot{\bs{\sigma}}_\rho &= \varphi_\rho(t,\bs{\rho}),\\
        \dot{\bs{\xi}}_1 &= \bs{\xi}_2, \\
        \dot{\bs{\xi}}_2 &= f(\bs{\xi},\bs{\theta}_1,\dot{\bar{\bs{\theta}}}_r) + G(\bs{\theta}_1)\bs{\tau} + \bs{\varsigma}_\xi, \\
        \dot{\bs{\varsigma}}_\xi &= \varphi_\xi(t,\bs{\xi}),\\
        \dot{\bs{x}}_{c_1} &= \bs{x}_{c_2}, \\
        \dot{\bs{x}}_{c_2} &= \bs{x}_{c_3}, \\
        \dot{\bs{x}}_{c_3} &= \bs{\sigma}_{x c}, \\
        \dot{\bs{\sigma}}_{x c} &= \varphi_{x c}(t,\bs{x}_c),\\
    \end{split}
\end{equation}
where $ \bs{\sigma}_{xc} =\frac{\partial \bar{f}_c(t,\bs{x}_c)}{\partial \bs{x}_c}\dot{\bs{x}}_c$. Since the reference system dynamics may not be known, they have been replaced by the disturbance term in their entirety. The estimated reference system states will be taken as the reference trajectory.

We now define the state vectors 
\begin{equation*}
    \begin{gathered}
        \bs{q} = [\bs{\rho}_1^\top \; \bs{\rho}_2^\top \; \bs{\xi}_1^\top \; \bs{\xi}_2^\top]^\top, \quad \bs{\chi}_1 = [\bs{\rho}_1^\top \; \bs{\rho}_2^\top \; \bs{\sigma}_\rho^\top]^\top, \\
        \bs{\chi}_2 = [\bs{\xi}_1^\top \; \bs{\xi}_2^\top \; \bs{\varsigma}_\xi^\top]^\top, \quad \bs{\chi}_3 = [\bs{x}_{c_1}^\top \; \bs{x}_{c_2}^\top \; \bs{x}_{c_3}^\top \;
        \bs{\sigma}_{x c}^\top]^\top,\\
        \bs{\chi} = [\bs{\chi}_1^\top \; \bs{\chi}_2^\top \; \bs{\chi}_3^\top]^\top.
    \end{gathered}
\end{equation*}
Define $\varphi(t,\bs{q},\bs{x}_c) = \left[\varphi_\rho(t,\bs{\rho}) \ \varphi_\xi(t,\bs{\xi}) \ \varphi_{x c}(t,\bs{x}_c)\right]^\top$, which is a vector of unknown functions describing the disturbance. It is assumed $\varphi(t,\bs{q},\bs{x}_c)$ is continuous and bounded on any compact set in $\bs{q}$ and $\bs{x}_c$. Note that the second order derivative of the reference trajectory, $\ddot{\bs{\theta}}_r$, is lumped into the disturbance $\bs{\varsigma}_\xi$. Since the disturbance term must satisfy \textit{Assumption \ref{as:disturbance}}, $\ddot{\bs{\theta}}_r$ must be differentiable, therefore requiring the translational reference signal to be fifth order differentiable.
The observer system with extended states can be written compactly as
\begin{equation*}\label{eq:combinedDynamics}
    \begin{split}
        \dot{\hat{\bs{\chi}}} &= A\hat{\bs{\chi}} + B\left[\bar{f}(\hat{\bs{\chi}},\bs{\theta}_1,\dot{\bar{\bs{\theta}}}_r) + \bar{G}(\bs{\theta}_1)\hat{\bs{\omega}}_s\right] + H\hat{\bs{\chi}}_e, \\
        \hat{\bs{\chi}}_e &= C(\bs{\chi} - \hat{\bs{\chi}}),
    \end{split}
\end{equation*}
where
\begin{equation*}
    \begin{gathered}
        A = \text{diag}\left(A_i\right), \ B = \text{diag}\left(B_i\right), \\ 
        C = \text{diag}\left(C_i\right), \ H = \text{diag}\left(H_i\right),
    \end{gathered}
\end{equation*}
\begin{equation*}
    \begin{gathered}
        A_i = \left[\begin{smallmatrix} 0_3 & I_3 & 0_3 \\ 0_3 & 0_3 & I_3 \\ 0_3 & 0_3 & 0_3\end{smallmatrix}\right], \ B_i = \left[\begin{smallmatrix} 0_3\\ I_3 \\ 0_3\end{smallmatrix}\right], \ H_i = \left[\begin{smallmatrix} \alpha_1/\epsilon I_3 \\ \alpha_2/\epsilon^2 I_3 \\ \alpha_3/\epsilon^3 I_3 \end{smallmatrix}\right], \\
        C_i = \left[\begin{smallmatrix} I_3 & 0_3 & 0_3\end{smallmatrix}\right], \ \text{for} \ i \in \{1,2\}, 
    \end{gathered}
\end{equation*}
\begin{equation*}
    \begin{gathered}
        A_3 = \left[\begin{smallmatrix} 0_3 & I_3 & 0_3 & 0_3\\ 0_3 & 0_3 & I_3 & 0_3\\ 0_3 & 0_3 & 0_3 & I_3 \\ 0_3 & 0_3 & 0_3 & 0_3\end{smallmatrix}\right], \ B_3 = \left[\begin{smallmatrix} 0_3 \\ 0_3 \\ 0_3 \\ 0_3\end{smallmatrix}\right], \ H_3 = \left[\begin{smallmatrix} \alpha_1/\epsilon I_3 \\ \alpha_2/\epsilon^2 I_3 \\ \alpha_3/\epsilon^3 I_3 \\ \alpha_4/\epsilon^4 I_3\end{smallmatrix}\right], \\
        C_3 = \left[\begin{smallmatrix} I_3 & 0_3 & 0_3 & 0_3\end{smallmatrix}\right],
    \end{gathered}
\end{equation*}
\begin{equation*}
    \bar{f}(\hat{\bs{\chi}},\bs{\theta}_1,\dot{\bar{\bs{\theta}}}_r) = \left[\begin{smallmatrix} g\bs{e}_z - \ddot{\bs{p}}_r \\ f(\hat{\bs{\xi}},\bs{\theta}_1,\dot{\bar{\bs{\theta}}}_r) \\ 0_{3\times1}\end{smallmatrix}\right], \ \bar{G}(\bs{\theta}_1) = b\left[\begin{smallmatrix} \frac{-R_3(\bs{\theta}_1)}{m} & 0_{3} \\ 0_{3\times1} & G(\bs{\theta}_1) \\ 0_{3\times1} & 0_3 \end{smallmatrix}\right]M,
\end{equation*}
where $H$ is designed by choosing $\alpha_j^i$ such that
\begin{equation}
    s^{\varrho_i} + \alpha_1^i s^{\varrho_i-1} + \dots + \alpha_{\varrho_i-1}^i s + \alpha_{\varrho_i},
\end{equation}
is Hurwitz and $[\varrho_1 \ \varrho_2 \ \varrho_3]^\top = [3 \ 3 \ 4]^\top$.

\subsection{Output Feedback Control}
The state feedback controllers \eqref{eq:rotationalControl} and \eqref{eq:translationalControl} are rewritten as output feedback controllers
\begin{equation}\label{eq:rotationalOutputFeedback}
    \hat{\bs{\tau}}_d = G^{-1}(\bs{\theta}_1)\left[\hat{\bs{f}}_r - f(\hat{\bs{\xi}},\bs{\theta}_1,\dot{\bar{\bs{\theta}}}_r)\right],
\end{equation}
where $\hat{\bs{f}}_r = -\beta_1\hat{\bs{\xi}}_1 - \beta_2\hat{\bs{\xi}}_2 - \hat{\bs{\varsigma}}_\xi$ and
\begin{equation}\label{eq:translationalOutputFeedback}
    \begin{split}
        \hat{\phi}_r &= \tan^{-1}\left( \dfrac{-\hat{f}_y}{\sqrt{\hat{f}_x^2 + (\hat{f}_z-g)^2}} \right), \quad \hat{\psi}_r = 0, \\
        \hat{\theta}_r &= \tan^{-1}\left(\dfrac{\hat{f}_x}{\hat{f}_z-g}\right), \quad \hat{u}_{f d} = -\dfrac{m(\hat{f}_z-g)}{c_{\hat{\phi}_r} c_{\hat{\theta}_r}},
    \end{split}
\end{equation}
where $\hat{\bs{f}}_t$ is defined using tracking error estimates
\begin{equation}\label{eq:translationalForcingOF}
    \hat{\bs{f}}_t = -\gamma_1\hat{\bs{\rho}}_1 - \gamma_2\hat{\bs{\rho}}_2 - \hat{\bs{\sigma}}_\rho + \hat{\bs{x}}_{c_3}.
\end{equation}
To be used for output feedback control, the estimates must be saturated to overcome the peaking phenomenon, see \textit{Remark \ref{rmk:peaking}}. The following saturation function is used to saturate each estimate individually
\begin{equation*}
    \hat{\chi}_{is} = k_{\chi_i}\operatorname{sat}\left(\frac{\hat{\chi}_i}{k_{\chi_i}} \right), \ \  \operatorname{sat}(y)=\left\{\begin{array}{ll}{y,} & {\text { if }|y| \leq 1}, \\ {\operatorname{sign}(y),} & {\text { if }|y|>1},\end{array}\right.
\end{equation*}
for $1 \leq i \leq 30$, where the saturation bounds $k_{\chi_i}$ are chosen such that the saturation functions will not be invoked under state feedback. These saturated estimates are used in the output feedback controllers \eqref{eq:rotationalOutputFeedback} and \eqref{eq:translationalOutputFeedback}.

\section{Stability Analysis}\label{sec:stabilityAnalysis}
The domain of operation will now be restricted and the stability of the state feedback control, observer estimates, and output feedback control will now be proven.
\subsection{Restricting Domain of Operation}
In order to ensure that the rotational feedback linearizing control law remains well defined, the domain of operation must be restricted leading to the following assumption
\begin{assumption}\label{as:rotationalRefSet}
The rotational reference signals remain in the set $\{|\phi_r| < \frac{\pi}{2} - \delta, |\theta_r| < \frac{\pi}{2} - \delta \}$, where $0<\delta<\frac{\pi}{2}$.
\end{assumption}

We will now establish that for sufficiently small initial tracking error, $\bs{\xi}(0)$, the tracking error $\norm{\bs{\xi}_1(t)}<\delta$ for all $t>0$. Consequently the system will operate away from any Euler angle singularities. A Lyapunov function in the rotational error dynamics is taken as
\begin{equation}\label{eq:v_xi}
    V_\xi = \bs{\xi}^\top P_\xi\bs{\xi}, \quad \text{where} \ P_\xi A_\xi + A_\xi^\top P_\xi = -I_6,
\end{equation}
\begin{equation*}
    A_\xi = \left[\begin{smallmatrix} 0_3 & I_3 \\ -\beta_1I_3 & -\beta_2I_3 \end{smallmatrix}\right].
\end{equation*}
A Lyapunov function in the translational error dynamics is taken as
\begin{equation}\label{eq:v_rho}
    V_\rho = \bs{\rho}^\top P_\rho\bs{\rho}, \quad \text{where} \ P_\rho A_\rho + A_\rho^\top P_\rho = -I_6,
\end{equation}
\begin{equation*}
    A_\rho = \left[\begin{smallmatrix} 0_3 & I_3 \\ -\gamma_1I_3 & -\gamma_2I_3 \end{smallmatrix}\right].
\end{equation*}
Define the following two positive constants $c_\rho,\rho_\text{max}\inr_{>0}$. Let the positive constant $c_\xi\inr_{>0}$ be chosen such that $c_\xi < (\beta_1 + 1)\delta^2/(2\beta_2)$.
\begin{lemma}[Restricting Domain of Operation]\label{lem:restrictingDomainOfAnalysis}
    For the feedback linearized rotational error dynamics \eqref{eq:rotationalClosedLoopStateFeedback} with initial conditions $\bs{\xi}(0)$ in the set $\Omega_\xi = \{V_\xi \leq c_\xi\}$ the system state $\bs{\xi}(t)$ remains in the set $\norm{\bs{\xi}_1(t)}<\delta$ for all $t>0$. Similarly, the feedback linearized translational error dynamics \eqref{eq:translationalClosedLoopStateFeedback} with initial conditions $\bs{\rho}(0)$ in the set $\Omega_\rho = \{V_\rho \leq c_\rho\}$ the system state $\bs{\rho}(t)$ remains in the set $\norm{\bs{\rho}(t)}\leq\rho_{\max}$ for all $t>0$.
\end{lemma}

\begin{proof}
Solving $P_\xi A_\xi + A_\xi^\top P_\xi = -I_6$ for $P_\xi$, the Lyapunov function \eqref{eq:v_xi} is
\begin{equation*}
    \small
    V_\xi = \frac{(\beta_1 + 1)\bs{\xi}_1^\top \bs{\xi}_1}{2\beta_2} + \frac{\beta_1\bs{\xi}_2^\top \bs{\xi}_2 + (\beta_2\bs{\xi}_1 + \bs{\xi}_2)^\top(\beta_2\bs{\xi}_1 + \bs{\xi}_2)}{2\beta_1\beta_2}.
\end{equation*}
Taking the bound on the Lyapunov function
\begin{equation*}
    V_\xi \leq c_\xi \Rightarrow \frac{(\beta_1 + 1)\bs{\xi}_1^\top\bs{\xi}_1}{2\beta_2} \leq c_\xi,
\end{equation*}
and choosing $c_\xi$ in the following manner
\begin{equation*}
    c_\xi < \frac{(\beta_1 + 1)\delta^2}{2\beta_2} \Rightarrow \norm{\bs{\xi}_1(t)} < \delta,
\end{equation*}
over the set $\Omega_\xi$. The Lyapunov function \eqref{eq:v_xi} also satisfies the following inequalities
\begin{equation*}
    \lambda_{\min} (P_\xi)\norm{\bs{\xi}}^2 \leq V_\xi \leq \lambda_{\max} (P_\xi)\norm{\bs{\xi}}^2, \quad \dot{V}_\xi \leq -\norm{\bs{\xi}}^2,
\end{equation*}
showing that $\Omega_\xi$ is positively invariant, which ensures Euler angle singularities will be avoided so long as the reference signals of the rotational subsystem remain in the set of \textit{Assumption \ref{as:rotationalRefSet}}.

In view of potential rotational tracking errors, the Lyapunov function \eqref{eq:v_rho} satisfies the following inequalities
\begin{equation*}
    \begin{gathered}
        \lambda_{\min} (P_\rho)\norm{\bs{\rho}}^2 \leq V_\rho \leq \lambda_{\max} (P_\rho)\norm{\bs{\rho}}^2,\\
        \dot{V}_\rho \leq -\norm{\bs{\rho}}^2 + 2[0 \; e_\theta(t,\bs{\xi}_1)^\top]^\top P_\rho \bs{\rho}.
    \end{gathered}
\end{equation*}
Since $e_\theta(t,\bs{\xi}_1)$ and its partial derivatives are continuous on $\Omega_\xi$, and $e_\theta$ is uniformly bounded in time, $e_\theta$ is locally Lipschitz in $\bs{\xi}_1$ on $\Omega_\xi$. We can now define
\begin{equation*}
    \norm{e_\theta(t,\bs{\xi}_1) - e_\theta(t,0)} \leq L_e\norm{\bs{\xi}_1} \leq L_e \delta,
\end{equation*}
for the Lipschitz constant, $L_e$. We can then bound the Lyapunov function derivative by
\begin{equation*}
    \dot{V}_\rho \leq -\norm{\bs{\rho}}^2 + 2 L_e \delta \lambda_\text{max}(P_\rho)\norm{\bs{\rho}}.
\end{equation*}
Define $\rho_\text{max} = 2 L_e \delta \lambda_\text{max}(P_\rho)$, for $\norm{\bs{\rho}} > \rho_\text{max}$, $\dot{V}_\rho < 0$. Since $V_\rho < \lambda_\text{max}(P_\rho)\norm{\bs{\rho}}^2$ we can choose
\begin{equation}
    c_\rho = \lambda_\text{max}(P_\rho)(2 L_e \delta \lambda_\text{max}(P_\rho))^2.
\end{equation}
By this choice, $\norm{\bs{\rho}(t)} \leq \rho_\text{max}$ for all $t>0$ and $\Omega_\rho$ is compact and positively invariant. The domain of operation is now defined as the set $\Omega_q = \Omega_\xi \times \Omega_\rho$.
\end{proof}

\begin{corollary}\label{cor:rotationalSet}
By \textit{Lemma \ref{lem:restrictingDomainOfAnalysis}} and \textit{Assumption \ref{as:rotationalRefSet}}, the rotational states remain in the set $\{|\phi| < \frac{\pi}{2}, |\theta| < \frac{\pi}{2}, |\dot{\phi}| < a_\theta, |\dot{\theta}| < a_\theta \}$, where $a_\theta$ is some positive constant. Thereby ensuring singularities in the Euler angles are avoided and the feedback linearizing controller \eqref{eq:rotationalControl} remains well defined.
\end{corollary}

\subsection{Stability Under State Feedback}\label{subsec:stabilityUnderStateFeedback}
\begin{theorem}[Stability Under State Feedback]\label{thm:SFStability}
    For the closed-loop state feedback rotational and translational subsystems,  \eqref{eq:rotationalClosedLoopStateFeedback} and \eqref{eq:translationalClosedLoopStateFeedback}, in the presence of rotational tracking error, i.e., $\bs{e}_\theta(t,\bs{\xi}_1) \neq 0$, whose initial conditions $(\bs{\xi}(0),\bs{\rho}(0))$ start in $\Omega_q$, the system states $\bs{\xi}(t)$ and $\bs{\rho}(t)$ will remain in $\Omega_q$ for all $t>0$. The states will exponentially converge to the origin.
\end{theorem}
\begin{proof}
The translational and rotational closed-loop systems can be written as a cascaded system in the following form
\begin{equation*}\label{eq:cascadedSystem}
    \begin{aligned}[c]
        \dot{\bs{\rho}}_1 &= \bs{\rho}_2, \\
        \dot{\bs{\rho}}_2 &= -\gamma_1\bs{\rho}_1 - \gamma_2\bs{\rho}_2 + \bs{e}_\theta(t,\bs{\xi}_1),\\
        \dot{\bs{\xi}}_1 &= \bs{\xi}_2,\\
        \dot{\bs{\xi}}_2 &= -\beta_1\bs{\xi}_1 - \beta_2\bs{\xi}_2,
    \end{aligned}
    \ \Rightarrow \
    \begin{aligned}
        \dot{\bs{\rho}}_1 &= \bs{\rho}_2,\\
        \dot{\bs{\rho}}_2 &= f_1(t,\bs{\rho},\bs{\xi}),\\
        \dot{\bs{\xi}}_1 &= \bs{\xi}_2,\\
        \dot{\bs{\xi}}_2 &= f_2(\bs{\xi}).
    \end{aligned}
\end{equation*}
Taking the Lyapunov functions for the rotational and translational subsystems, \eqref{eq:v_xi} and \eqref{eq:v_rho}, a composite Lyapunov function can be written
\begin{equation}
    V_{sf} = d_1V_\rho + V_\xi, \quad d_1>0.
\end{equation}
It can be shown following the generalized proof in the Appendix that for $d_1$ small enough, the entire closed-loop state feedback system converges exponentially to the origin, and the set $\Omega_q$ is compact and positively invariant.
\end{proof}

\subsection{Convergence of Observer Estimates}\label{subsec:convergenceOfObserverEstimates}
The scaled error dynamics of the EHGO are written by making the following change of variables
\begin{equation}
    \begin{gathered}
        \bs{\eta}_j^i = \frac{(\bs{\chi}_j^i - \hat{\bs{\chi}}_j^i)}{\epsilon^{\varrho_i-j}}, \quad \tilde{\bs{\omega}}_s = \bs{\omega}_s - \hat{\bs{\omega}}_s,
    \end{gathered}
\end{equation}
where $\chi_j^i$ is the $j$-th element of $\bs{\chi}_i$ for $1\leq i \leq 3$ and $1\leq j \leq \varrho_i$, ${\hat{\chi}_j^i}$ is the estimate of $\chi_j^i$ obtained using the EGHO and $\hat{\bs{\omega}}_s = \left[\hat{\omega}_1^2 \dots \hat{\omega}_n^2 \right]^\top$. In the new variables, the scaled EHGO estimation error dynamics become
\begin{equation}\label{eq:fullObserverErrorDynamics}
    \begin{split}
        \epsilon\dot{\bs{\eta}}^i &= F_i\bs{\eta}^i + B_1^i\left[\Delta \bar{f}^i + \bar{G}^i(\bs{\theta}_1)\tilde{\bs{\omega}}_s\right]\\
        &\quad + \epsilon B_2^i\left[\varphi^i(t,\bs{q},\bs{x}_c)\right],
    \end{split}
\end{equation}
where $\Delta \bar{f}^i = \bar{f}^i(\bs{\chi},\bs{\theta}_1,\dot{\bar{\bs{\theta}}}_r) - \bar{f}^i(\hat{\bs{\chi}},\bs{\theta}_1,\dot{\bar{\bs{\theta}}}_r)$ and $\bar{f}^i$, $\bar{G}^i$, and $\varphi^i$ correspond to rows $3i-2$ to $3i$ of $\bar{f}$, $\bar{G}$, and $\varphi$ respectively. Note \eqref{eq:fullObserverErrorDynamics} is an $O(\epsilon)$ perturbation of
\begin{equation}\label{eq:simpleEHGOSystem}
    \epsilon\dot{\bs{\eta}}^i = F_i\bs{\eta}^i + B_1^i\left[\Delta \bar{f}^i + \bar{G}^i(\bs{\theta}_1)\tilde{\bs{\omega}}_s\right],
\end{equation}
where
\begin{equation*}
    F_{i}=\left[\begin{smallmatrix}{-\alpha_{1}^{i}}I_3 & {I_3} & {\cdots} & {0_3} \\ {\vdots} & {} & {\ddots} & {\vdots} \\ {-\alpha^i_{\varrho_i-1}}I_3 & {0_3} & {\cdots} & {I_3} \\ {-\alpha^i_{\varrho_i}}I_3 & {0_3} & {\cdots} & {0_3}\end{smallmatrix}\right], \quad B_1^i = \left[\begin{smallmatrix} 0_3 \\ \vdots \\ I_3 \\ 0_3 \end{smallmatrix}\right],
\end{equation*}
\begin{equation*}
    B_2^i = \left[0_3 \ \cdots \ 0_3 \ I_3\right]^\top, \quad \bs{\eta}^i = \left[{\bs{\eta}_1^i}^\top \ \cdots \ {\bs{\eta}_{\varrho_i}^i}^\top\right]^\top.
\end{equation*}

The actuator error dynamics in terms of the error in squared rotor angular rate, $\tilde{\bs{\omega}}_s$, and rotor angular rate error, $\tilde{\bs{\omega}} = \bs{\omega} - \hat{\bs{\omega}}$ can be written as
\begin{equation}\label{eq:actuatorErrorDynamics}
    \begin{split}
        \tau_m\dot{\tilde{\bs{\omega}}}_s &= -2\tilde{\bs{\omega}}_s + 2W_\text{des}\tilde{\bs{\omega}}, \\
        \tau_m\dot{\tilde{\bs{\omega}}} &= -\tilde{\bs{\omega}}.
    \end{split}
\end{equation}
Although $W_\text{des} = \text{diag}[\bs{\omega}_\text{des}] \inr^{n\times n}$ is time-varying, by exploiting the fact that $W_\text{des}$ is bounded, the actuator error dynamics can be analyzed as a cascaded system.
\begin{lemma}[Stability of Actuator Dynamics] 
    For bounded input $\bs{\omega}_{\operatorname{des}}$, i.e., $\omega_{\operatorname{des}_i} \leq \omega_{\max}$ for $i\in\{1,\dots,n\}$, where $\omega_{\max}$ is some positive constant, the actuator error dynamics \eqref{eq:actuatorErrorDynamics} exponentially converge to the origin.
\end{lemma}
\begin{proof}
The actuator error dynamics are a cascaded system with the Lyapunov functions
\begin{equation}
    V_{\tilde{\omega}_s} = \tilde{\bs{\omega}}_s^\top\tilde{\bs{\omega}}_s, \quad V_{\tilde{\omega}} = \tilde{\bs{\omega}}^\top\tilde{\bs{\omega}},
\end{equation}
and the composite Lyapunov function
\begin{equation}\label{eq:v_nu}
    V_v = d_2V_{\tilde{\omega}_s} + V_{\tilde{\omega}}, \quad d_2>0.
\end{equation}
Using the general result for cascaded systems in the Appendix, it can be shown that the origin is exponentially stable where $d_2$ is chosen small enough.
\end{proof}

Combining these states as $\bs{v} = [\tilde{\bs{\omega}}_s^\top \ \tilde{\bs{\omega}}^\top]^\top$, the entire actuator error dynamic system can be written as
\begin{equation}\label{eq:actuatorErrorDynamicsCondensed}
    \dot{\bs{v}} = A_\omega\bs{v}, \quad \text{where} \ A_\omega = \left[\begin{matrix} -2I_n & 2W_\text{des} \\ 0_n & -I_n \end{matrix}\right].
\end{equation}
The systems \eqref{eq:simpleEHGOSystem} and \eqref{eq:actuatorErrorDynamicsCondensed} can be cascaded as
\begin{equation}\label{eq:EHGOcascadeSystem}
    \begin{split}
        \epsilon\dot{\bs{\eta}}^i &= F_i\bs{\eta}^i + B_1^i\left[\Delta \bar{f}^i + \left[\bar{G}^i(\bs{\theta}_1) \ 0_{9\times n}\right]\bs{v}\right], \\
        \tau_m\dot{\bs{v}} &= A_\omega \bs{v}.
    \end{split}
\end{equation}
\begin{theorem}[Convergence of EHGO Estimates]\label{thm:EHGOConvergence}
    The estimates of the EHGOs, $\hat{\bs{\chi}}$, in \eqref{eq:simpleEHGOSystem} converge exponentially to the states, $\bs{\chi}$, i.e., $\bs{\eta}$ converges exponentially to the origin for any $\epsilon \in (0,\epsilon^*)$ for sufficiently small $\epsilon^*>0$.
\end{theorem}

\begin{proof}
Using cascade analysis, a composite Lyapunov function for the cascaded system \eqref{eq:EHGOcascadeSystem} is designed. The Lyapunov function for the actuator error system is \eqref{eq:v_nu} and a Lyapunov function for the EHGO error system with the input, $\bs{v}$, set to zero is written as
\begin{equation}
     \epsilon V_\eta = \sum_{i=1}^3(\bs{\eta}^i)^\top P_3^i\bs{\eta}^i, \quad P_3^i F_i + F_i^\top P_3^i = -I,
\end{equation}
leading to
\begin{equation*}
    \epsilon \dot{V}_\eta = \sum_{i=1}^3-(\bs{\eta}^i)^\top\bs{\eta}^i + 2(\bs{\eta}^i)^\top P_3^i B_1^i\Delta \bar{f}^i.
\end{equation*}

The function $\bar{f}^i(\bs{\chi},\bs{\theta}_1,\dot{\bar{\bs{\theta}}}_r)$ is a function of time through $\bs{\theta}_1$ and $\dot{\bar{\bs{\theta}}}_r$, is uniformly bounded in time, and $\bar{f}^i$ and its partial derivatives are continuous. Therefore, it is Lipschitz in $\bs{\chi}$ on $\Omega_q$ and $\Delta \bar{f}^i$ can be bounded by
\begin{equation*}
    \norm{\Delta \bar{f}^i} \leq L_\eta\norm{\bs{\chi}_i - \hat{\bs{\chi}}_i} \Rightarrow \norm{\Delta \bar{f}^i} \leq \epsilon L_\eta\norm{\bs{\eta}^i},
\end{equation*}
leading to the following bound on the derivative of the Lyapunov function
\begin{equation*}
    \begin{split}
        \epsilon \dot{V}_\eta &\leq \sum_{i=1}^3\left(-\norm{\bs{\eta}^i}^2 + 2\epsilon L_\eta\norm{\bs{\eta}^i}^2\norm{P_3^i B_1^i}\right), \\
        \epsilon \dot{V}_\eta &\leq -\norm{\bs{\eta}}^2 + 2\epsilon L_\eta\norm{N}\norm{\bs{\eta}}^2,
    \end{split}
\end{equation*}
where the elements of the diagonal matrix $N$ are $N_i = \norm{P_3^i B_1^i}$. Since $\alpha_j^i$ are tunable and $\epsilon$ is a design parameter, pick $\epsilon$ such that $2\epsilon L_\eta \norm{N} \leq \frac{1}{2}$ resulting in the following inequality
\begin{equation}
    \epsilon \dot{V}_\eta \leq -\frac{1}{2}\norm{\bs{\eta}}^2.
\end{equation}
Taking the composite Lyapunov function as 
\begin{equation}
    V_{\hat{\chi}} = d_3V_\eta + V_v, \quad d_3>0,
\end{equation}
and following the Appendix, the origin of \eqref{eq:EHGOcascadeSystem} is exponentially stable.
\end{proof}


The complete scaled observer error system \eqref{eq:fullObserverErrorDynamics} is the same as \eqref{eq:simpleEHGOSystem} with an added perturbation. The perturbation is bounded by $\epsilon\varphi(t,\bs{q},\bs{x}_c)<\epsilon \kappa$ for some  $\kappa \inr_{>0}$ and is continuous and bounded, therefore it can be treated as a nonvanishing perturbation. Following \textit{Lemma 9.2} in \cite{khalil2002nonlinear} and \textit{Theorem 2}, the estimation error of the EHGO converges exponentially to an $O(\epsilon)$ neighborhood of the origin. 
\begin{remark}[Peaking Phenomenon]\label{rmk:peaking} 
The EHGO estimation error $\tilde{\chi}_i = \chi_i - \hat{\chi}_i$ can be bounded by
\begin{equation}
    |\tilde{\chi}_i| \leq \frac{b}{\epsilon^{\varrho - 1}}\norm{\bs{\tilde{\chi}}(0)}e^{-at/\epsilon},
\end{equation}
for some positive constants $a$ and $b$, by \textit{Theorem 2.1} in \cite{khalil2017HGO}. Initially, the estimation error can be very large, i.e., $O(1/\epsilon^{\varrho-1})$, but will decay rapidly. To prevent the peaking of the estimates from entering the plant during the initial transient, the output feedback controller needs to be saturated outside a compact set of interest. This is done by saturating the individual estimates and results in the output feedback controller  $\bs{u}(\hat{\bs{\chi}}_s)$.


There is some set $\{V_\eta\leq \epsilon^2c\}$ for some $c\inr_{>0}$ that the estimation error will enter after some short time, $T(\epsilon)$, where $\lim_{\epsilon \rightarrow 0}T(\epsilon) = 0$. Since the initial state $\bs{q}(0)$ resides on the interior of $\Omega_q$, choosing $\epsilon$ small enough will ensure that $\bs{q}$ will not leave the set $\Omega_q$ during the interval $[0,T(\epsilon)]$. This establishes the boundedness of all states.
\end{remark}

\subsection{Stability Under Output Feedback}
The system under output feedback is a singularly perturbed system which can be split into two time-scales. The multi-rotor dynamics and control reside in the slow time-scale while the observer and actuator dynamics reside in the fast time-scale.
\begin{theorem}[Stability Under Output Feedback]\label{OFStability}
    The closed-loop system under output feedback, with initial conditions $\bs{q}(0)$ on the interior of $\Omega_q$ will exponentially converge to the origin when $\epsilon$ is chosen small enough.
\end{theorem}

\begin{proof}
The entire output feedback closed-loop system can now be written in singularly perturbed form
\begin{subequations}\label{eq:singularlyPerturbedForm}
\begin{equation}\label{eq:slowSystem}
        \dot{\bs{q}} = A_c\bs{q} + O(\epsilon), \quad \quad \quad \quad \quad \quad \quad \quad \
\end{equation}
\begin{equation}\label{eq:boundaryLayerSystem}
    \begin{split}
        \epsilon\dot{\bs{\eta}}^i &= F_i\bs{\eta}^i + B_1^i\left[\Delta \bar{f}^i + \left[\bar{G}^i(\bs{\theta}_1) \ 0_{9\times n}\right]\bs{v}\right], \\
        \tau_m\dot{\bs{v}} &= A_\omega \bs{v},
    \end{split}
\end{equation}
\end{subequations}
where
\begin{equation*}
    A_c = \left[\begin{smallmatrix} A_\rho & 0_6 \\ 0_6 & A_\xi \end{smallmatrix}\right],
\end{equation*}
and the $O(\epsilon)$ perturbation in \eqref{eq:slowSystem} is due to any estimation error.
This closed-loop system has a two-time-scale structure because $\epsilon$ and $\tau_m$ are small. Since the effect of the $O(\epsilon)$ perturbation in \eqref{eq:fullObserverErrorDynamics} vanishes as $\epsilon$ is pushed to zero, the boundary layer system can be taken as \eqref{eq:boundaryLayerSystem}. The slow dynamics can be taken as \eqref{eq:slowSystem}. From \textit{Theorem \ref{thm:EHGOConvergence}}, the origin of the boundary layer system is an exponentially stable equilibrium point, and from \textit{Theorem \ref{thm:SFStability}}, the origin of the slow system is an exponentially stable equilibrium point. Following \textit{Theorem 11.4} in \cite{khalil2002nonlinear},
it can be shown that the entire closed-loop system with output feedback control has an exponentially stable equilibrium point at the origin. 
\end{proof}

\section{Simulation}\label{sec:simulation}
The control algorithm is implemented where the reference system is taken as a ground vehicle on which the multi-rotor will land. The initial position of the multi-rotor is $\bs{p}_1(0) = [1, \ 1, \ -4]^\top$ and the initial position of the ground vehicle is $\bs{x}_{c_1}(0) = [5, \ 0, \ -0.5]^\top$. The ground vehicle follows a figure-8 trajectory and the multi-rotor is tasked with tracking and landing on the ground vehicle. While only having a position measurement of the ground vehicle, with added noise, the multi-rotor is able to land on the vehicle, as shown in Fig. \ref{fig:simulationResults}. The multi-rotor is able to make this landing while canceling disturbances of $\bs{\sigma}_\xi = [\sin(t) \ \cos(t) \ \sin(t)]^\top$, and $\bs{\sigma}_\rho = [\cos(t) \ \sin(t) \ \cos(t)]^\top$ with Gaussian white noise added to all measurement signals.
\begin{figure}
    \centering
    \includegraphics[width=80mm]{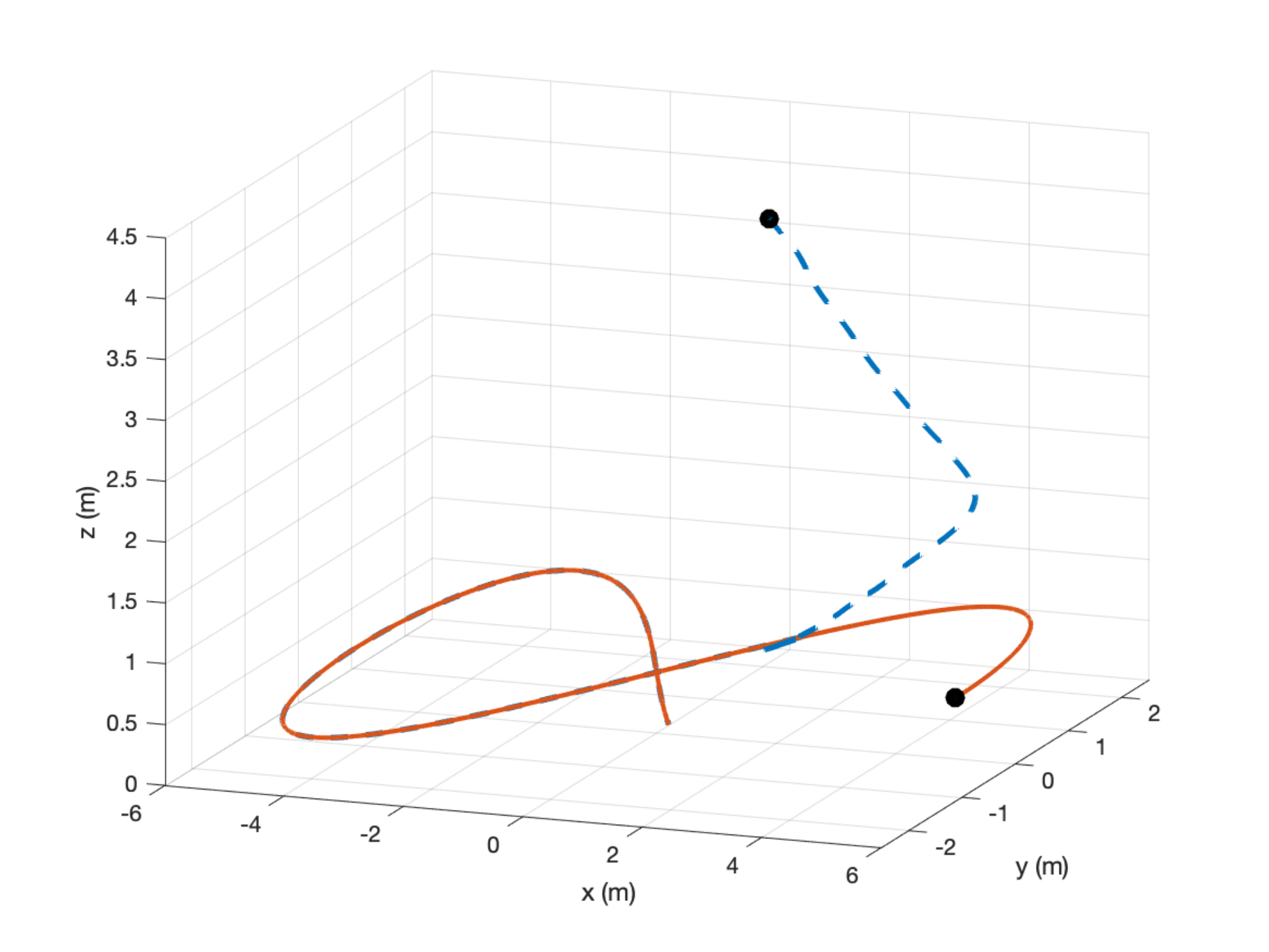}
    \caption{The trajectory of the ground vehicle (solid) and the trajectory of the multi-rotor UAV (dashed).}
    \label{fig:simulationResults}
\end{figure}

\section{Conclusion and Future Directions}\label{sec:conclusion}
We studied an unknown trajectory tracking problem for a multi-rotor with modeling error and external disturbances. The unknown trajectory is generated from a dynamical system with unknown or partially known dynamics. We designed and rigorously analyzed an EHGO-based output feedback controller. We illustrated the controller using the example of the multi-rotor landing on a mobile platform. 

We plan to extend this work in several possible directions. First, we are conducting experiments to characterize the efficacy of the proposed controller in an actual test-bed. Second, we plan to extend the proposed trajectory tracking framework to allow for actuator failures. 


\section*{Appendix}
\bit{Stability of Generalized Cascade System:}
A generalized stability proof for cascaded systems is adapted from Appendix C.1 of \cite{khalil2015nonlinear}. Consider the cascade connection of two systems
\begin{equation}\label{eq:genericCascadedSystem}
    \dot{\eta} = f_1(t,\eta,\xi), \quad \dot{\xi} = f_2(\xi),
\end{equation}
where $f_1$ and $f_2$ are locally Lipschitz and $f_1(t,0,0)=0$, $f_2(0)=0$. Assuming the origin of $\dot{\xi} = f_2(\xi)$ is exponentially stable, there is a continuously differentiable Lyapunov function, $V_2(\xi)$, that satisfies the following inequalities
\begin{subequations}
\begin{equation}
    c_1\norm{\xi}^2 \leq V_2(\xi) \leq c_2\norm{\xi}^2,
\end{equation}
\begin{equation}\label{eq:V4Inequalities}
    \quad \frac{\partial V_2(\xi)}{\partial \xi}f_2(\xi) \leq -c_3\norm{\xi}^2,
\end{equation}
\begin{equation}
    \norm{\frac{\partial V_2(\xi)}{\partial \xi}} \leq c_4\norm{\xi}.
\end{equation}
\end{subequations}
Now, suppose there is a continuously differentiable Lyapunov function that satisfies the inequalities
\begin{equation}
    \frac{\partial V_1}{\partial \eta}f_1(t,\eta,0) \leq -c\norm{\eta}^2, \quad \norm{\frac{\partial V_1}{\partial \eta}} \leq k\norm{\eta}.
\end{equation}
Take a composite Lyapunov function for the cascaded system as
\begin{equation}
    V(\eta,\xi) = b V_1(\eta) + V_2(\xi), \quad b > 0,
\end{equation}
in which $b$ can be chosen. The derivative $\dot{V}$ satisfies
\begin{equation*}
    \begin{split}
        \dot{V}(\eta,\xi) &= b\frac{\partial V_1(\eta)}{\partial \eta}f_1(t,\eta,0) + \\
        &b\frac{\partial V_1(\eta)}{\partial \eta}\left[f_1(t,\eta,\xi) - f_1(t,\eta,0)\right] + \frac{\partial V_2(\xi)}{\partial \xi}f_2(\xi),\\
        \dot{V}(\eta,\xi) &\leq -b c\norm{\eta}^2 + b k L\norm{\eta}\norm{\xi} - c_3\norm{\xi}^2,
    \end{split}
\end{equation*}
where $f_1$ and its partial derivatives are continuous on the set of interest and $f_1$ is uniformly bounded in time, therefore it is locally Lipschitz. Here $L$ is a Lipschitz constant of $f_1$ with respect to $\xi$. The inequality can be written in quadratic form as
\begin{equation*}
    \begin{split}
        \dot{V} &\leq -\left[\begin{matrix} \norm{\eta} \\ \norm{\xi}\end{matrix}\right]^\top \left[\begin{matrix} b c & \frac{-b k L}{2}\\ \frac{-b k L}{2} & c_3 \end{matrix}\right] \left[\begin{matrix} \norm{\eta} \\ \norm{\xi}\end{matrix}\right], \\
        &= -\left[\begin{matrix}\norm{\eta} \\ \norm{\xi}\end{matrix}\right]^\top Q \left[\begin{matrix} \norm{\eta} \\ \norm{\xi}\end{matrix}\right] \leq -\lambda_{min}(Q)\norm{\left[\begin{matrix} \norm{\eta} \\ \norm{\xi}\end{matrix}\right]}^2,
    \end{split}
\end{equation*}
where $b$ is chosen such that $b < 4cc_3/(k L)^2$ to ensure $Q$ is positive definite. The foregoing analysis shows that the origin of \eqref{eq:genericCascadedSystem} is exponentially stable.

\end{document}